\documentclass[aps,pra,twocolumn,notitlepage,superscriptaddress,showpacs,nofootinbib]{revtex4-1}
\usepackage{amsmath}
\usepackage{amssymb}
\usepackage{amsfonts}
\usepackage{enumerate}
\usepackage{mathrsfs}
\usepackage{graphicx}
\usepackage{epstopdf}
\usepackage{enumerate}
\usepackage{changepage}
\usepackage{bm}
\usepackage{multirow}
\usepackage{lipsum}
\usepackage{booktabs}

\usepackage{tikz}
\usepackage{caption}
\usepackage{subcaption}

\newtheorem{theorem}{Theorem}
\newtheorem{definition}{Definition}

\newtheorem{lemma}{Lemma}
\newtheorem{proposition}{Proposition}
\newtheorem{conjecture}{Conjecture}
\newtheorem{example}{Example}
\newtheorem{corollary}{Corollary}

\def\bcj{\begin{conjecture}}
	\def\ecj{\end{conjecture}}
\def\bcr{\begin{corollary}}
	\def\ecr{\end{corollary}}
\def\bd{\begin{definition}}
	\def\ed{\end{definition}}
\def\bea{\begin{eqnarray}}
	\def\eea{\end{eqnarray}}
\def\bem{\begin{enumerate}}
	\def\eem{\end{enumerate}}
\def\bex{\begin{example}}
	\def\eex{\end{example}}
\def\bim{\begin{itemize}}
	\def\eim{\end{itemize}}
\def\bl{\begin{lemma}}
	\def\el{\end{lemma}}
\def\bma{\begin{bmatrix}}
	\def\ema{\end{bmatrix}}
\def\bpf{\begin{proof}}
	\def\epf{\end{proof}}
\def\bpp{\begin{proposition}}
	\def\epp{\end{proposition}}
\def\bqu{\begin{question}}
	\def\equ{\end{question}}
\def\br{\begin{remark}}
	\def\er{\end{remark}}
\def\bt{\begin{theorem}}
	\def\et{\end{theorem}}

\def\squareforqed{\hbox{\rlap{$\sqcap$}$\sqcup$}}
\def\qed{\ifmmode\squareforqed\else{\unskip\nobreak\hfil
		\penalty50\hskip1em\null\nobreak\hfil\squareforqed
		\parfillskip=0pt\finalhyphendemerits=0\endgraf}\fi}
\def\endenv{\ifmmode\;\else{\unskip\nobreak\hfil
		\penalty50\hskip1em\null\nobreak\hfil\;
		\parfillskip=0pt\finalhyphendemerits=0\endgraf}\fi}
\newenvironment{proof}{\noindent \textbf{{Proof.~} }}{\qed}
\def\Dbar{\leavevmode\lower.6ex\hbox to 0pt
	{\hskip-.23ex\accent"16\hss}D}
\makeatletter
\def\url@leostyle{%
	\@ifundefined{selectfont}{\def\UrlFont{\sf}}{\def\UrlFont{\small\ttfamily}}}
\makeatother
\def\bcj{\begin{conjecture}}
	\def\ecj{\end{conjecture}}
\def\bcr{\begin{corollary}}
	\def\ecr{\end{corollary}}
\def\bd{\begin{definition}}
	\def\ed{\end{definition}}
\def\bea{\begin{eqnarray}}
	\def\eea{\end{eqnarray}}
\def\bem{\begin{enumerate}}
	\def\eem{\end{enumerate}}
\def\bex{\begin{example}}
	\def\eex{\end{example}}
\def\bim{\begin{itemize}}
	\def\eim{\end{itemize}}
\def\bl{\begin{lemma}}
	\def\el{\end{lemma}}
\def\bpf{\begin{proof}}
	\def\epf{\end{proof}}
\def\bpp{\begin{proposition}}
	\def\epp{\end{proposition}}
\def\bqu{\begin{question}}
	\def\equ{\end{question}}
\def\br{\begin{remark}}
	\def\er{\end{remark}}
\def\bt{\begin{theorem}}
	\def\et{\end{theorem}}

\def\btb{\begin{tabular}}
	\def\etb{\end{tabular}}

	\newcommand{\nc}{\newcommand}
	
	\nc{\bbA}{\mathbb{A}} \nc{\bbB}{\mathbb{B}} \nc{\bbC}{\mathbb{C}}
	\nc{\bbD}{\mathbb{D}} \nc{\bbE}{\mathbb{E}} \nc{\bbF}{\mathbb{F}}
	\nc{\bbG}{\mathbb{G}} \nc{\bbH}{\mathbb{H}} \nc{\bbI}{\mathbb{I}}
	\nc{\bbJ}{\mathbb{J}} \nc{\bbK}{\mathbb{K}} \nc{\bbL}{\mathbb{L}}
	\nc{\bbM}{\mathbb{M}} \nc{\bbN}{\mathbb{N}} \nc{\bbO}{\mathbb{O}}
	\nc{\bbP}{\mathbb{P}} \nc{\bbQ}{\mathbb{Q}} \nc{\bbR}{\mathbb{R}}
	\nc{\bbS}{\mathbb{S}} \nc{\bbT}{\mathbb{T}} \nc{\bbU}{\mathbb{U}}
	\nc{\bbV}{\mathbb{V}} \nc{\bbW}{\mathbb{W}} \nc{\bbX}{\mathbb{X}}
	\nc{\bbZ}{\mathbb{Z}}
	
	\nc{\bA}{{\bf A}} \nc{\bB}{{\bf B}} \nc{\bC}{{\bf C}}
	\nc{\bD}{{\bf D}} \nc{\bE}{{\bf E}} \nc{\bF}{{\bf F}}
	\nc{\bG}{{\bf G}} \nc{\bH}{{\bf H}} \nc{\bI}{{\bf I}}
	\nc{\bJ}{{\bf J}} \nc{\bK}{{\bf K}} \nc{\bL}{{\bf L}}
	\nc{\bM}{{\bf M}} \nc{\bN}{{\bf N}} \nc{\bO}{{\bf O}}
	\nc{\bP}{{\bf P}} \nc{\bQ}{{\bf Q}} \nc{\bR}{{\bf R}}
	\nc{\bS}{{\bf S}} \nc{\bT}{{\bf T}} \nc{\bU}{{\bf U}}
	\nc{\bV}{{\bf V}} \nc{\bW}{{\bf W}} \nc{\bX}{{\bf X}}
	\nc{\ba}{{\bf a}} \nc{\be}{{\bf e}} \nc{\bu}{{\bf u}}
	\nc{\brr}{{\bf r}} \nc{\bx}{{\bf x}}
	
	\nc{\cA}{{\cal A}} \nc{\cB}{{\cal B}} \nc{\cC}{{\cal C}}
	\nc{\cD}{{\cal D}} \nc{\cE}{{\cal E}} \nc{\cF}{{\cal F}}
	\nc{\cG}{{\cal G}} \nc{\cH}{{\cal H}} \nc{\cI}{{\cal I}}
	\nc{\cJ}{{\cal J}} \nc{\cK}{{\cal K}} \nc{\cL}{{\cal L}}
	\nc{\cM}{{\cal M}} \nc{\cN}{{\cal N}} \nc{\cO}{{\cal O}}
	\nc{\cP}{{\cal P}} \nc{\cQ}{{\cal Q}} \nc{\cR}{{\cal R}}
	\nc{\cS}{{\cal S}} \nc{\cT}{{\cal T}} \nc{\cU}{{\cal U}}
	\nc{\cV}{{\cal V}} \nc{\cW}{{\cal W}} \nc{\cX}{{\cal X}}
	\nc{\cZ}{{\cal Z}}
	
	\nc{\hA}{{\hat{A}}} \nc{\hB}{{\hat{B}}} \nc{\hC}{{\hat{C}}}
	\nc{\hD}{{\hat{D}}} \nc{\hE}{{\hat{E}}} \nc{\hF}{{\hat{F}}}
	\nc{\hG}{{\hat{G}}} \nc{\hH}{{\hat{H}}} \nc{\hI}{{\hat{I}}}
	\nc{\hJ}{{\hat{J}}} \nc{\hK}{{\hat{K}}} \nc{\hL}{{\hat{L}}}
	\nc{\hM}{{\hat{M}}} \nc{\hN}{{\hat{N}}} \nc{\hO}{{\hat{O}}}
	\nc{\hP}{{\hat{P}}} \nc{\hR}{{\hat{R}}} \nc{\hS}{{\hat{S}}}
	\nc{\hT}{{\hat{T}}} \nc{\hU}{{\hat{U}}} \nc{\hV}{{\hat{V}}}
	\nc{\hW}{{\hat{W}}} \nc{\hX}{{\hat{X}}} \nc{\hZ}{{\hat{Z}}}
	
	\nc{\hn}{{\hat{n}}}

	\newcommand{\ket}[1]{|#1\rangle}

	\newcommand{\braket}[2]{\langle#1|#2\rangle}

	\newcommand{\fl}[2]{\left\lfloor\frac{#1}{#2}\right\rfloor}
	\newcommand{\fc}[2]{\left\lceil\frac{#1}{#2}\right\rceil}

	\usepackage[
	colorlinks,
	linkcolor = blue,
	citecolor = blue,
	urlcolor = blue]{hyperref}
	\def \qed {\hfill \vrule height7pt width 7pt depth 0pt}
	
	\newcounter{lastnote}

\begin{document}
		\title{Unextendible product bases from orthogonality graphs}
	\author{Fei Shi}
\affiliation{QICI Quantum Information and Computation Initiative, Department of Computer Science,
The University of Hong Kong, Pokfulam Road, Hong Kong}	

\author{Ge Bai}
	\affiliation{QICI Quantum Information and Computation Initiative, Department of Computer Science,
The University of Hong Kong, Pokfulam Road, Hong Kong}

\author{Xiande Zhang}
\affiliation{School of Mathematical Sciences,
	University of Science and Technology of China, Hefei, 230026,  China}
 \affiliation{Hefei National Laboratory, University of Science and Technology of China, Hefei 230088, People’s Republic of China}
		
\author{Qi Zhao}
\affiliation{QICI Quantum Information and Computation Initiative, Department of Computer Science,
The University of Hong Kong, Pokfulam Road, Hong Kong}		

\author{Giulio Chiribella}
\affiliation{QICI Quantum Information and Computation Initiative, Department of Computer Science,
The University of Hong Kong, Pokfulam Road, Hong Kong}	 
\affiliation{Department of Computer Science, Parks Road, Oxford, OX1 3QD, United Kingdom}	
\affiliation{Perimeter Institute for Theoretical Physics, Waterloo, Ontario N2L 2Y5, Canada}	
		

		\begin{abstract}
  Unextendible product bases (UPBs) play a key role in the study of quantum entanglement and nonlocality.    A famous open question  is whether there exist  genuinely unextendible product bases (GUPBs), namely  multipartite product bases that are unextendible with respect to every possible bipartition.  Here we shed light on this question by providing a characterization of UPBs and GUPBs in terms of orthogonality graphs.  Building on this connection, we develop a method for constructing UPBs in low dimensions, and we derive a lower bound on the size of any GUPB, significantly improving over the state of the art. Moreover, we  show that every minimal GUPB saturating our bound must be associated to regular  graphs. Finally,  we discuss  a possible path towards the construction of a minimal GUPB  in a tripartite system of minimal local dimension.		\end{abstract}
	\maketitle

		\section{Introduction}\label{sec:int}

An important notion in the study of quantum entanglement and nonlocality is the notion of   unextendible product basis (UPB) \cite{bennett1999unextendible}.    Mathematically, a UPB is a set of orthogonal product vectors whose complementary subspace contains no product vector \cite{bennett1999unextendible}.  UPBs have   a number  of properties that make them  important in   quantum information and quantum foundations. For example, the complementary subspace of a UPB is a completely entangled subspace, that is, a subspace containing only  entangled states  \cite{parthasarathy2004maximal,bhat2006completely,walgate2008generic}. The normalized projector on the complementary subspace of a UPB is a bound entangled state, that is, a state from which  no pure entanglement can be distilled \cite{bennett1999unextendible,divincenzo2003unextendible}. Another important property is that the states in a UPB  cannot be perfectly distinguished using local operations and classical communication,  a phenomenon that has become known as {\em quantum nonlocality without entanglement} \cite{bennett1999unextendible,bennett1999quantum} and has been recently shown to admit a device-independent certification  \cite{vsupic2022self}.  UPBs also play a central role in the study of Bell inequalities with no quantum violation \cite{augusiak2011bell, Augusiak2012tight,fritz2013local,acin2016guess}, where they offer  insights into the foundations of quantum theory.  
    
The construction and characterization of UPBs  has attracted great  attention over the past two decades \cite{divincenzo2003unextendible,AL01,Fen06,Joh13,Johnston2014The,Chen2013The,chen2019unextendible,chen2018no,chen2018multiqubit,halder2019family,shi2020unextendible,shi2021strong,shi2022strongly}. A famous open question  in the field is whether there exists a multipartite UPB that is a UPB with respect to  every possible bipartition.  Such a UPB is called a {\em genuinely unextendible product basis (GUPB)} \cite{demianowicz2018unextendible} and its complementary subspace  is a genuinely entangled subspace, that is, a subspace that contains only genuinely entangled states \cite{parthasarathy2004maximal,cubitt2008dimension,demianowicz2018unextendible}. 

Sets of orthogonal  product states  that cannot be completed to full product bases in every bipartion were found in Ref.~\cite{shi2022unextendible}. However, these sets do not provide examples of GUPBs, because non-completability to a full product basis for the whole Hilbert space  is a weaker property  than non-extendibility to a larger set of orthogonal product states.   A universal construction for genuinely entangled subspaces  was given in \cite{demianowicz2022universal}.    However, determining whether the orthogonal complement of such subspaces admits a product basis, and in the affirmative case, constructing the product  basis is highly non-trivial. For these reasons, the  existence of GUPBs is still an open question.

Recently, Demianowicz 
gave a lower bound on size that GUPBs must have, if they exist \cite{PhysRevA.106.012442}: for an $N$-partite  GUPB  in $\bbC^{d_1}\otimes \bbC^{d_2}\otimes \cdots \otimes \bbC^{d_N}$ the number of vectors in the basis, denoted by $k$, must satisfy the bound  
\begin{align}\label{old} 
k\geq  \frac{D}{d_{\text{max}}}+\fl{\frac{D}{d_{\text{max}}}-2}{N-1}+1 \,,
\end{align}
where $D:=d_1d_2\cdots d_N$, and $d_{\text{max}}:=\text{max}\{d_1,d_2,\ldots, d_N\}$  (here and in the rest of the paper, we always assume the condition  $d_m  \ge 3$ for every $m\in \{1,\dots, N\}$ because no bipartite UPB---and therefore no GUPB---exists when one of  the local dimensions  is smaller  than 3  \cite{bennett1999unextendible,divincenzo2003unextendible}).

In this paper, we provide a characterization of UPBs and GUPBs  in terms of   orthogonality graphs \cite{lovasz1989orthogonal}, a central notion in classical and quantum information theory \cite{lovasz1979shannon,duan2012zero,chiribella2013confusability,ramanathan2014necessary,cabello2014graph,duan2015no,duan2016zero,wang2018separation,wang2017semidefinite}.    Our characterization translates directly into  a constructive method for building UPBs, which we illustrate by building a new UPB for a two-qubits and two-qutrits quantum system.  For GUPBs, the characterization implies a new lower bound, which significantly improves over the state of the art. Specifically, we show that the size of a GUPB   in $\bbC^{d_1}\otimes \bbC^{d_2}\otimes \cdots \otimes \bbC^{d_N}$ is lower bounded as
\begin{align}\label{new}
k\geq  \frac{\sum_{m=1}^N\frac{D}{d_m}-1}{N-1} \,.
\end{align}
In general, the estimate of $k$ provided by Eq.~\eqref{new} is always larger than or equal to the estimate of $k$ provided by Eq.~\eqref{old}. The difference between the two bounds becomes visible when the component  systems  have different local dimensions. For example,  consider a tripartite system where the component systems have local dimensions  $d_1= d_2 =2p$ and $d_3  =  3p$ for some integer $p$. In this case, bounds (\ref{old}) and (\ref{new}) read $k\ge 6p^2$ and $k\ge 8 p^2$, respectively, and the difference  between them becomes arbitrarily large as $p$ increases.

The connection between UPBs/GUPBs and orthogonality graphs  also implies other  constraints on the structure of UPBs/GUPBs.   In particular,  we show that minimal UPBs saturating a bound by Bennett {\em et al.} \cite{bennett1999unextendible} must necessarily correspond to regular graphs, and we show that the same holds for   minimal GUPBs saturating our bound (\ref{new}).     Finally, we use the regularity condition to discuss a possible  path  to the construction of  a minimal GUPB  in a tripartite quantum system of minimal local dimension.




The rest of this paper is organized as follows. In Sec.~\ref{sec:pri}, we review the concepts of UPBs, GUPBs, and orthogonality graphs. In Sec.~\ref{sec:upbs}, we establish a connection between UPBs and orthogonality graphs,  and derive upper and lower bounds on the degrees of vertices of the orthogonality graphs associated to UPBs.  In Sec.~\ref{sec:bounds}, we derive Eq.~\eqref{new} and discuss its relations with other bounds on the size of GUPBs. In Sec.~\ref{sec:improved_bound}, we provide an improved bound valid for certain local dimensions.
In Sec.~\ref{sec:gupbs}, we show that  minimal GUPBs saturating the bound (\ref{new}) should be associated to regular graphs, and we use this result to discuss a possible route to construct a minimal GUPB. Finally, the conclusions are provided  in Sec.~\ref{sec:con}.


		\section{Preliminaries}\label{sec:pri}

In this section, we review a few basic facts about notation, unextendible product bases, orthogonality graphs, and orthogonal representations. 
\medskip 

{\em Notation.}  In this paper, the number of vectors in a UPB  will always be denoted by $k$ and will be called the {\em size} of the UPB.  The total dimension of the space $\bbC^{d_1}\otimes \bbC^{d_2}\otimes\cdots\otimes \bbC^{d_N}$ 
will always be denoted  by $D:=d_1d_2\cdots d_N$.    Moreover, we will  assume that the local dimensions are listed in non-decreasing order, namely $d_1\leq d_2\leq \cdots\leq d_N$.   Finally,  we will often work with unnormalized product states, which simplifies some of the expressions.

\medskip

{\em Unextendible product bases.}     Let us start from the mathematical definition:  

\begin{definition}
A set of orthogonal product states   $\cU=\{\ket{\varphi_1^{(i)}}_{A_1} \ket{\varphi_2^{(i)}}_{A_2} \cdots \ket{\varphi_N^{(i)}}_{A_N}  \}_{i=1}^k \subset \bbC^{d_1}\otimes \bbC^{d_2}\otimes\cdots\otimes \bbC^{d_N}$   is an unextendible product basis (UPB) if the orthogonal complement of ${\sf Span}\{\ket{\varphi_1^{(i)}}_{A_1} \ket{\varphi_2^{(i)}}_{A_2} \cdots \ket{\varphi_N^{(i)}}_{A_N}  \}_{i=1}^k$    has non-zero dimension and contains no product state. A UPB is called  a genuinely unextendible product basis (GUPB) if it is  a UPB with respect to every possible bipartition of the tensor product $\bbC^{d_1}\otimes \bbC^{d_2}\otimes\cdots\otimes \bbC^{d_N}$.
\end{definition}	

A well-known result about bipartite UPBs is that they can only exist if the local dimensions are strictly larger than 2 \cite{bennett1999unextendible,divincenzo2003unextendible}:  in other words, there is no UPB for bipartite systems of the form  $\bbC^2\otimes \bbC^n$ or of the form $\bbC^n\otimes \bbC^2$, for some  $n\geq 2$. 

In the multipartite case, it is important to stress that the notion of GUPB is  much stronger than the notion of multipartite UPB. A multipartite UPB cannot be extended by any vector of the fully-product form $\ket{\psi_1}_{A_1} \ket{\psi_2}_{A_2} \cdots \ket{\psi_N}_{A_N}$, where $\ket{\psi_m}_{A_m}$ is a state of subsystem $A_m$.    In contrast,  a GUPB cannot even be extended by vectors of the form $\ket{\Psi_1}_{S_1}  \otimes  \ket{\Psi_2}_{S_2}$ where $\ket{\Psi_1}$ and $\ket{\Psi_2}$ are (possibly entangled) states of the quantum systems associated to a partition of the composite system $A_1 \cdots A_N$ into two disjoint parts  $S_1$ and $S_2$.  

The fact that no bipartite UPB can exist with local dimensions smaller than 3 implies that an $N$-partite GUPB can only exist if  
\begin{align}
d_m  \ge  3,  \qquad \forall m\in\{1,
\dots,  N\}.
\end{align}

Another important  type of constraint on  multipartite UPBs and GUPBs concerns their size.    A   first  bound was provided by Bennett and coauthors~\cite{bennett1999unextendible}, who showed  that the size of a multipartite UPB is lower bounded as  
\begin{align}\label{UPBbound} 
k\geq \sum_{m=1}^N(d_m-1)+1 \,.
\end{align}
Later, Alon and Lov\'asz \cite{AL01} showed that the above inequality holds with the $``>"$ sign if     at least one of the dimensions $(d_m)_{m=1}^N$ is  even  and the sum  $\sum_{m=1}^N(d_m-1)+1$ is odd.
  Applying the above bounds to the bipartition $(A_1  |  A_2  \cdots A_N)$  yields the following bounds on the size of  GUPBs \cite{PhysRevA.106.012442} 
  \begin{equation}\label{trivial}
  \begin{aligned}
k  \ge 
\left\{ 
\begin{array}{ll} 
 d_1+\frac{D}{d_1}, \qquad & {\rm if}~d_1~{\rm and}~\frac D{d_1}~{\rm are~even};~ \\
 &\\
 d_1+\frac{D}{d_1}-1, & {\rm otherwise}.
 \end{array} 
 \right.
 \end{aligned}
 \end{equation}
In the rest of the paper, we will call a  lower bound {\em non-trivial} if it improves over Eq. (\ref{trivial})  for some values of $N$ and of the local dimensions.    An example of a non-trivial lower bound is  Demianowicz's bound (\ref{old}) in the case when  $(N-1)  d_N  < N \,  d_1$  and when  certain conditions on the local dimensions are satisfied \cite{PhysRevA.106.012442}.   Another example of a non-trivial lower bound is our bound (\ref{new}), which is non-trivial for a larger set of values of the local dimensions. 

\begin{figure*}[t]
		\centering
		\includegraphics[scale=0.6]{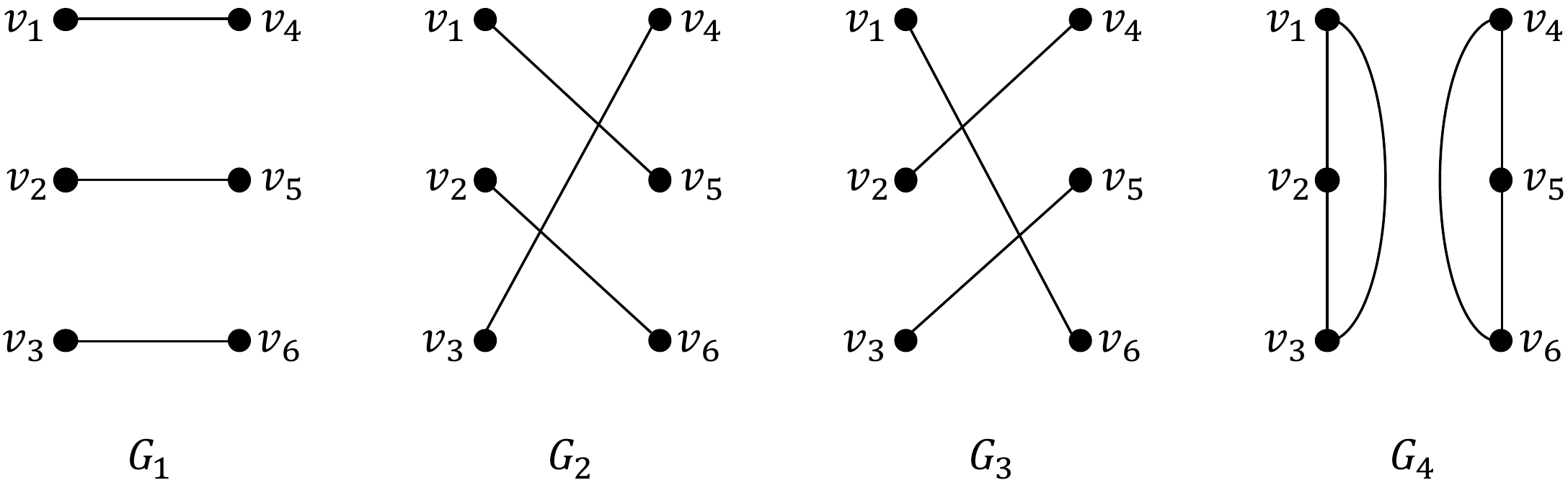}
		\caption{Orthogonality graphs of the UPB  in Example \ref{example:upb2223}. } \label{fig:orthogonality_graph}
	\end{figure*}

\medskip

\emph{Orthogonality graphs}. An undirected simple graph $G=(V,E)$ is an ordered pair consisting of a set $V$ of vertices, and a set $E$ of edges, which is an irreflexive, symmetric relation on $V$. A vertex $u$ is a neighbor of a vertex $v$ if $u$ and  $v$ are adjacent, namely $(u,v)\in E$.    The neighborhood $N_G(v)$ of a vertex $v$ is the set of all neighbors of $v$. The degree ${\rm deg}_G(v)$ is the  number of vertices in  the neighborhood $N_G(v)$, i.e. ${\rm deg}_G(v)=|N_G(v)|$.  If the degree of each vertex is $k$, the graph is called  $k$-regular.   A complete graph $K_n$ is an $(n-1)$-regular  graph with $n$ vertices, that is, a graph in which every two different vertices are connected. For two graphs $G_1=(V_1,E_1)$ and $G_2=(V_2,E_2)$, the union of graphs $G_1$ and $G_2$ is the graph $G_1\cup G_2=(V_1\cup V_2,E_1\cup E_2)$.

The {\em orthogonality graph} of a set of vectors  $\{  |\varphi^{(1)} \rangle, \dots, |\varphi^{(k)}\rangle\} \subset \bbC^d$, is the graph $G  =  (V,  E)$ with  vertex set $V=\{v_1, \cdots,v_k\}$ and edge set  
$E=\{(v_i, v_j)\mid \braket{\varphi^{(i)}}{\varphi^{(j)}}=0\}$.  

A connection between UPBs and orthogonality graphs was made by Alon and Lov\'asz in Ref. \cite{AL01}, where it was used to prove existence results about minimal UPBs satisfying Bennett {\em el al.'s}  bound (\ref{UPBbound}). We now introduce a new definition that will allow us to provide an if and only if characterization of UPBs in terms of orthogonality graphs. 

\begin{definition}
Let $G = (V,E)$ be the  orthogonality graph of the set $\{  |\varphi^{(1)} \rangle, \dots, |\varphi^{(k)}\rangle\} \subset \bbC^d$, and let $W  \subseteq V$ be a subset of the vertices.    We say that the subset  $W $ is {\em saturated} if the corresponding vectors  $\{  |\varphi^{(i)}  \rangle\mid   v_i  \in  W\}$ span the whole space $\bbC^d$.   Otherwise, we call the set $W$ {\em unsaturated.}     
\end{definition}

For a set of $N$-partite product vectors 
$\{\ket{\varphi_1^{(i)}}_{A_1} \ket{\varphi_2^{(i)}}_{A_2} \cdots \ket{\varphi_N^{(i)}}_{A_N}\}_{i=1}^k$ in $\bbC^{d_1}\otimes \bbC^{d_2}\otimes\cdots\otimes \bbC^{d_N}$, one can  define $N$ orthogonality graphs: 
\begin{definition}   Let $\{\ket{\varphi_1^{(i)}}_{A_1} \ket{\varphi_2^{(i)}}_{A_2} \cdots \ket{\varphi_N^{(i)}}_{A_N}\}_{i=1}^k$ in $\bbC^{d_1}\otimes \bbC^{d_2}\otimes\cdots\otimes \bbC^{d_N}$ be a set of $N$-partite product vectors.  For $m  \in  \{1,\dots, N\}$, the  orthogonality graph $G_m=(V,E_m)$  is the graph with  vertex set $V=\{v_1,v_2,\cdots,v_k\}$ and edge set  $E_m=\{(v_i, v_j)\mid \braket{\varphi_m^{(i)}}{\varphi_m^{(j)}}_{A_m}=0\}$. 
\end{definition}

Note that all the graphs $G_m$ have the same vertex set, and (generally) different edges due to the (generally) different orthogonality relations between the  vectors in different subsystems.

We now give a necessary and sufficient condition, formulated in terms of orthogonality graphs, for a set of product states to be a  UPB.

\begin{lemma}\label{lemma:suf_nec}
Let   $\cU$ be a set of $k$ product vectors in $\bbC^{d_1}\otimes \bbC^{d_2}\otimes\cdots\otimes \bbC^{d_N}$, and let $(G_m)_{m=1}^N$ be the corresponding orthogonality graphs.   The set $\cU$ is a UPB if and only if the following conditions hold:
\begin{enumerate}[(i)]
    \item $\bigcup_{m=1}^NG_m=K_k$;
    \item  $\bigcup_{m=1}^NW_m\neq V$ for every $N$-tuple  $(W_1 ,W_2,\ldots,W_N) $ in which  $W_m$ is an unsaturated set for  $G_m$ for  every $m\in  \{1,
    \dots, N\}$.  \end{enumerate}
\end{lemma}
 
 

The proof of Lemma~\ref{lemma:suf_nec} is provided in Appendix~\ref{appendix:lemma}. To illustrate the lemma, we consider the following example:  

\begin{example}\label{example:upb2223}
The following product vectors   form a UPB in  $\bbC^2\otimes \bbC^2\otimes \bbC^2\otimes \bbC^3$:
\begin{small}
\begin{equation}\label{eq:upb2223}
    \begin{aligned}
    \ket{\psi_1}=&\ket{0}_{A_1} \ket{0}_{A_2} \ket{0}_{A_3} \ket{0}_{A_4},\\  
    \ket{\psi_2}=&(\ket{0}+\ket{1})_{A_1} (\ket{0}+\ket{1})_{A_2} (\ket{0}+\ket{1})_{A_3}\ket{1}_{A_4},  \\
    \ket{\psi_3}=&(\ket{0}+2\ket{1})_{A_1} (\ket{0}+2\ket{1})_{A_2} (\ket{0}+2\ket{1})_{A_3}\ket{2}_{A_4}, \\  
    \ket{\psi_4}=&\ket{1}_{A_1}(2\ket{0}-\ket{1})_{A_2}(\ket{0}-\ket{1})_{A_3}(\ket{0}+\ket{1}+\ket{2})_{A_4},\\
    \ket{\psi_5}=&(\ket{0}-\ket{1})_{A_1}\ket{1}_{A_2}(2\ket{0}-\ket{1})_{A_3}(\ket{0}+2\ket{1}-3\ket{2})_{A_4},\\ 
    \ket{\psi_6}=&(2\ket{0}-\ket{1})_{A_1}(\ket{0}-\ket{1})_{A_2}\ket{1}_{A_3}(5\ket{0}-4\ket{1}-\ket{2})_{A_4}.
    \end{aligned}
\end{equation}
\end{small}
The UPB $\{\ket{\psi_i}\}_{i=1}^6$ has  the minimum size compatible with Bennett {\em et al.}'s bound  (\ref{UPBbound}), which in this case reads $k  \ge \sum_{m=1}^4  (d_m -1)  +  1=  6$.  
\end{example}

\textbf{Analysis of Example \ref{example:upb2223}}. For the vectors in Example \ref{example:upb2223},  the  orthogonality graphs $(G_m)_{m=1}^4$ and their common vertex set are shown in Fig.~\ref{fig:orthogonality_graph}.   It is then easy to check that the union of the graphs $(G_m)_{m=1}^4$ is the complete graph $K_6$.  Hence, the first condition in Lemma~\ref{lemma:suf_nec} is satisfied. Regarding the second condition, note that every two vectors in  the set $\{\ket{0}, (\ket{0}+\ket{1}), (\ket{0}+2\ket{1}), \ket{1}, (\ket{0}-\ket{1}), (2\ket{0}-\ket{1})\}\subset \bbC^2$ are linearly independent, and therefore form a basis for $\bbC^2$.  Hence,  the size of any unsaturated set $W_m$ in $G_m$ can be at most 1 for every $m\in [1,3]$. Similarly, since any three vectors in the set  $\{\ket{0}, \ket{1}, \ket{2}, (\ket{0}+\ket{1}+\ket{2})_, (\ket{0}+2\ket{1}-3\ket{2}), (5\ket{0}-4\ket{1}-\ket{2})\}\subset \bbC^3$ are linearly independent,  the size of any unsaturated set $W_4$ in $G_4$ is at most 2. Putting everything together, we obtain that the union of  any four  unsaturated sets $(W_m)_{m=1}^4$ cannot contain all vertices in $V$. Since both conditions in Lemma~\ref{lemma:suf_nec} are satisfied,  we conclude that the vectors  $\{\ket{\psi_i}\}_{i=1}^6$ form a UPB. 

\medskip

Lemma~\ref{lemma:suf_nec}  also implies  an upper bound on the number of elements in the unsaturated sets  associated to a UPB:  
\begin{lemma}\label{lem:unsaturated}
 Let $(G_m)_{m=1}^N$ be the orthogonality graphs associated to a UPB of size $k$.  Then, the size of any unsaturated set $W_m$ in $G_m$   is upper bounded as 
 \begin{align}
 |W_m|  \le  k   -1-   \sum_{i\in  \{1,\dots,  N\}  \setminus \{m\} }  \,  (d_i-1) \,.
 \end{align}
 \end{lemma}
The proof is provided in Appendix \ref{appendix:lemma}.
\medskip 

{\em Orthogonal representations of a graph.} 
An \emph{orthogonal representation}  \cite{lovasz1989orthogonal}  of a graph $G=  (V,E)$  in dimension $d$ is a set of $k=|V|$  vectors $\{  |\varphi^{(1)} \rangle, \dots, |\varphi^{(k)}\rangle\} \subset \bbC^d$ such that  
 $\braket{\varphi^{(i)}}{\varphi^{(j)}}  = 0$  
for every pair of adjacent vertices $v_i$ and $v_j$.    The representation is called {\em faithful}  if $\braket{\varphi^{(i)}}{\varphi^{(j)}}  = 0$ only if $v_i$ and $v_j$ are adjacent.  

One way to search for an orthogonal representation of a given graph is to solve  the following optimization problem:
\begin{align}
    \text{minimize} &~ \sum_{(v_i,v_j)\in E} |\braket{\varphi^{(i)}}{\varphi^{(j)}}| \nonumber \\
    \text{subject to} & ~\braket{\psi}{\varphi^{(i)}} = 1,~\qquad \forall i=1,...,k \label{eq:min_st}
\end{align}
where $\ket{\varphi^{(i)}}$ are variable vectors and $\ket{\psi}$ is an arbitrarily chosen non-zero constant vector. Here, the constraint~\eqref{eq:min_st} ensures that every $\ket{\varphi^{(i)}}$ is non-zero. This constraint does not restrict the search space of orthogonal representations, since for any given valid orthogonal representation, one can always rotate the vectors globally so that each vector has a non-zero overlap with $\ket{\psi}$, and scale each vector individually to make every overlap be one. In our realization of this algorithm, we fix the first component of each vector $\ket{\varphi^{(i)}}$ to be one, which is equivalent to setting $\ket{\psi}$ to be the unit vector along the first axis. After optimization, if the objective function reaches zero, then the vectors $\{\ket{\varphi^{(i)}}\}_{i=1}^k$  satisfy the desired orthogonality relations.  

Eq.~\eqref{eq:min_st} shows that the search for orthogonal representations of a graph is an optimization problem with linear constraints. Various algorithms for this task are known such as Sequential Least Squares Programming \cite{han1976superlinearly,kraft1988software}.   
Notice that the problem in Eq.~\eqref{eq:min_st} is not a convex optimization, and therefore  optimization algorithms are not guaranteed to find the global minimum. Still, when an algorithm returns the value 0,  this value is automatically guaranteed to be the global minimum, and the result of the optimization is an orthogonal representation of the given graph.    In general, the solution may  not be faithful, meaning there may exist vectors $\ket{\varphi^{(i)}}$ and $\ket{\varphi^{(j)}}$ that are orthogonal even if the corresponding vertices  $v_i$ and $v_j$ are not adjacent.

Lemma~\ref{lemma:suf_nec}  suggests a systematic route to construct UPBs of any desired size $k$ in $\bbC^{d_1}\otimes \bbC^{d_2}\otimes\cdots\otimes \bbC^{d_N}$: 
\begin{enumerate}
\item decompose the complete graph $K_k$ into $N$ subgraphs $(G_m)_{m=1}^N$; 
\item  for each $G_m$,  find an orthogonal representation in $\bbC^{d_m}$; 
\item  for  each arbitrary  $N$-tuple  $(W_m)_{m=1}^N$ of unsaturated sets $W_m$  in $G_m =  (V,E_m)$, check that $\bigcup_{m=1}^NW_m\neq V$. 
\end{enumerate} 

The first step, namely the decomposition of the complete graph into subgraphs $(G_m)_{m=1}^N$ will be discussed in the next section of the paper.  Once a decomposition is given, the second step can be attempted by optimization algorithms that search for an orthogonal representation of the graphs $G_m$, as discussed in the previous paragraph.  Computationally, this step is the most challenging one.      
Finally, the third step can be achieved by brute-force enumerating all the unsaturated sets of the graphs $G_m$, once the decomposition $(G_m)_{m=1}^N$ and an orthogonal representation of the graphs $G_m$  are known.   The computational cost of this step is tolerable for instances of the problem where  $N$ and $D$ are small.  In general, the size of every unsaturated set  in $G_m$ is  upper bounded by $k-1 -  \sum_{i \in  \{1,\dots,  N\}\setminus\{m\}}(d_i-1)$ (by Lemma \ref{lem:unsaturated}), and further inspection of the structure of the orthogonality graphs $G_m$ and their orthogonal representations can further reduce this number.
Hence, as long as the number of systems $N$ and the total dimension $D$ are small, the enumeration of all $N$-tuples of unsaturated sets remains computationally feasible.   As an example of the application of our method, we constructed a new UPB of size $8$  for a two-qubits and two-qutrits system.   The basis and its orthogonality graphs are shown in  Example~\ref{example:upb2233} in Appendix~\ref{appendix:lemma}.   This UPB has  the minimum size compatible with Alon and Lov\'asz's bound  (\ref{UPBbound}), which in this case reads $ \sum_{m=1}^4  (d_m -1)  +  1=  7$ is odd, and at least one of the local dimensions $(2,2,3,3)$ is even.



Later in the paper, we will further discuss the minimal case $N=3$, with $d_1 =  d_2  = d_3  =  3$.  In this case,  $k$ can be generally bounded as $7\le k \le  23$ for UPBs and $ 13 \le k  \le   23$ for GUPBs (the lower bounds come from Eq. (\ref{UPBbound}) for UPBs and from Eqs. (\ref{old}) and (\ref{new}) for GUPBs, while the upper bound comes from the fact that the projector on the span of a UPB is the orthogonal complement of  a bound entangled state with positive partial transpose, and no such state  can have a rank smaller than 4  \cite{chen2013separability}).

  		\begin{table*}[t]
	\renewcommand\arraystretch{1.7}	
	\caption{ Comparison among three lower bounds on the GUPB size for different values of the local dimensions. 
 }\label{Table:compasions}
	\centering
	\renewcommand\tabcolsep{5pt}
	\begin{tabular}{c|c|c|c}
		\midrule[1.1pt]
		Local dimensions  &  Bound (\ref{old})  ~\cite{PhysRevA.106.012442}  	& Bound (\ref{trivial})~\cite{bennett1999unextendible,AL01,PhysRevA.106.012442}  & Our bound (\ref{new})      \\
		\midrule[1.1pt]
  	   (3,3,4)           &13    &14    &16            \\
	   (3,3,5)           &13    &17    &19            \\
	   (3,3,3,4)         &36    &38    &45            \\
	   (3,3,4,4)         &48    &50    &56           \\
	   (3,3,3,3,4)       &101   &110   &128           \\
	   (3,3,3,4,4)       &135   &146   &162          \\
		\midrule[1.1pt]
	\end{tabular}
\end{table*}

\section{Orthogonality graphs of UPBs}\label{sec:upbs}
In this section,  we show that the orthogonality graphs associated to UPBs must satisfy non-trivial  conditions on the degree of their  vertices.   In particular, we show that every minimal UPB saturating Bennett {\em et al.}'s bound (\ref{UPBbound}) must correspond to regular graphs.

\begin{lemma}\label{lem:deg}
For every  UPB in   $\bbC^{d_1}\otimes \bbC^{d_2}\otimes\cdots\otimes \bbC^{d_N}$, the degrees of the vertices in the  orthogonality graphs $(G_m)_{m=1}^N$ must satisfy the condition 
\begin{equation}\label{degreebounds}
\begin{aligned}
  &d_m-1\leq {\rm deg}_{G_m}(v_i)\leq k-1-\sum_{i \in  \{1,\dots,  N\} \setminus \{m\}}   (d_i  -  1), \\
  &\forall v_i\in  V \, ,\quad \forall  m\in  \{1,\dots,  N\} \,.
  \end{aligned}
\end{equation}

\end{lemma}


The proof of Lemma~\ref{lem:deg} is provided in Appendix~\ref{appendix:lemma}.


Our bounds on the degrees of the vertices are satisfied with the equality sign   when the UPB has the  minimal size compatible with Bennett {\em et al.}'s bound (\ref{UPBbound}):   

 
\begin{proposition}\label{pro:re}
For a minimal UPB saturating  Bennett {\em et al.}'s bound (\ref{UPBbound}), the orthogonality graph $G_m$ is a $(d_m-1)$-regular graph for every $m\in  \{1,\dots, N\}$.
\end{proposition}

An example of this situation is Example~\ref{example:upb2223}. There, $G_i$ is a $1$-regular graph for $1\leq i\leq 3$, and $G_4$ is a $2$-regular graph.


Since regularity is a strong graph-theoretic property,  Proposition \ref{pro:re} establishes strong constraint on every minimal UPB saturating Bennett {\em et al.}'s bound.    In the next section, we build on  the connection with orthogonality graphs to derive the bound (\ref{new})  on the size of candidate GUPBs.  Later in the paper, we will show that the bound (\ref{new}) plays for GUPBs a similar role as Bennett {\em et al.}'s bound for UPBs: as we will show, every minimal GUPB saturating bound (\ref{new}) must be associated to regular orthogonality graphs.   


\section{Bound on the  GUPB size}\label{sec:bounds}

In this section, we derive the bound (\ref{new}) and discuss its relations with other bounds on the size of GUPBs.

\begin{theorem}\label{thm:lower_bound}
Every GUPB    in $\bbC^{d_1}\otimes \bbC^{d_2}\otimes\cdots\otimes \bbC^{d_N}$ must satisfy the bound  (\ref{new}),  or equivalently 
\begin{equation}\label{eq:lower_bound}
    k\geq \fc{\sum_{m=1}^N\frac{D}{d_m}-1}{N-1}  \, . 
\end{equation}
\end{theorem}
\begin{proof}
Let us assume there exists a GUPB $\cU$ in $\bbC^{d_1}\otimes \bbC^{d_2}\otimes\cdots\otimes \bbC^{d_N}$. Since $\cU$ is  a UPB with respect to  the bipartition $A_m\mid  \{A_1A_2\cdots A_N\}\setminus\{A_m\}$ for $1\leq m\leq N$, Lemma~\ref{lem:deg} implies that  the degree of every vertex $v_i$ of $G_m$  satisfies the condition
\begin{equation}\label{eq:di}
    {\rm deg}_{G_{m}}(v_i)\leq k-1-(\frac{D}{d_m}-1)=k-\frac{D}{d_m} \, .
\end{equation}
Now,   Lemma \ref{lemma:suf_nec} tells us that the union of the graphs $(G_m)_{m=1}^N$ is the complete graph $K_k$.   Since  the complete graph $K_k$ is $(k-1)$-regular,  we have the bound
\begin{equation}\label{eq:sum}
     \sum_{m=1}^N{\rm deg}_{G_{m}}(v_i)\geq k-1.
\end{equation}
 Combining  Eqs.~\eqref{eq:di} and \eqref{eq:sum}, we  then obtain the relation 
  \begin{equation}\label{eq:sumd}
     \sum_{m=1}^N(k-\frac{D}{d_m})\geq k-1,
 \end{equation}
which implies the desired bound 
\begin{equation}\label{eq:lower_bound1}
    k\geq \frac{\sum_{m=1}^N\frac{D}{d_m}-1}{N-1}  \, . 
\end{equation}
Since $k$ is (by definition) an integer, the bound also holds with the ceiling sign, as in Eq. (\ref{eq:lower_bound}). 
\end{proof}
\vspace{0.4cm}

 Our lower bound coincides with    Demianowicz's bound (\ref{old}) when the local dimensions are all equal, {\em i.e.} if $d_m  =  d\, ,\forall m \in  \{1,\dots,  N\}$. 
 In general, however, our bound is strictly more accurate, as shown in the following proposition: 
 \begin{proposition}\label{prop:boundcomparison}
The r.h.s. of Eq.~\eqref{eq:lower_bound} is always larger than or equal to the r.h.s. of Eq.~\eqref{old}.
 \end{proposition}
 
The proof of Proposition~\ref{prop:boundcomparison} is provided in  Appendix \ref{appendix:proposition}.

\medskip 

 Another benefit of the new bound (\ref{new}) is that it  provides  non-trivial lower bound  in new cases, including  values of  the local dimensions for which no previous bound could improve over Eq.~\eqref{trivial}.     Some examples of this situation are illustrated  in  Table~\ref{Table:compasions}.

  In Ref.~\cite{PhysRevA.106.012442},  Demianowicz  showed that Eq.~\eqref{old} is a non-trivial lower bound  if and only if $(N-1)d_{\text{max}}<Nd_{\text{min}}$, where $d_{\text{min}}=\text{min}\{d_1,d_2,\ldots, d_N\}$, and the local dimensions satisfy the conditions  $(d_1,d_2,d_3)\neq (2p,2p,3p-1)$ and $(d_1,d_2,d_3)\neq (2p-1,\tilde{d},3p-2)$ for every integer  $p\geq 2$ and every  integer $\tilde d$ satisfying $2p-1\leq \tilde{d}\leq 3p-2$.  In contrast, we now show that our lower bound (\ref{new})  remains non-trivial even when the local dimensions are of the form  $(d_1,d_2,d_3)= (2p,2p,3p-1)$, or  $(d_1,d_2,d_3)= (2p-1,\tilde{d},3p-2)$. 
   
\begin{proposition}\label{cor:p}
In the tripartite case, the  bound  (\ref{new})  is  non-trivial when  
 $(d_1,d_2,d_3)= (2p,2p,3p-1)$ for some integer $p\ge 2$, and when $(d_1,d_2,d_3)= (2p-1,\tilde{d},3p-2)$ for some integer $p\geq 2$ and some integer $\tilde d  \in  [2p-1,  3p-2]$. 
\end{proposition}

 The proof of Proposition~\ref{cor:p} is provided in Appendix~\ref{appendix:proposition}.

\vspace{0.4cm}


\section{Improved bound under conditions on  the local dimensions}\label{sec:improved_bound}

We now show that our bound (\ref{new})  can be slightly  improved if the local dimensions satisfy certain conditions:  

\begin{proposition}\label{pro:shakinghand}
  If  at least one of the local dimensions $(d_m)_{m=1}^N$ is even and the sum $\sum_{m=1}^N\frac{D}{d_m}-1$ is an odd multiple of $N-1$, then  the size of any GUPB in $\bbC^{d_1}\otimes \bbC^{d_2}\otimes\cdots\otimes \bbC^{d_N}$ is lower bounded as
\begin{equation}
       k\geq \frac{\sum_{m=1}^N\frac{D}{d_m}-1}{N-1}+1.
\end{equation}
\end{proposition}

The proof of Proposition~\ref{pro:shakinghand} is provided in Appendix~\ref{appendix:proposition}.
For example, if $(d_1,d_2,d_3)=(3,4,5)$, then $k\geq 24$ for a GUPB of size $k$ in $\bbC^{3}\otimes \bbC^{4}\otimes \bbC^{5}$ by Proposition~\ref{pro:shakinghand}. Similarly, if $(d_1,d_2,d_3,d_4)=(4,4,4,4)$, then $k\geq 86$. This second example can be generalized to all situations in which the number of system $N$ is even and all local dimensions are equal to $N$:
\begin{corollary}
   If  $N$ is even and  $d_m  = N$ for every $m\in \{1,\dots,N  \}$,  then the minimum size of a GUPB in $\left(\bbC^{N}\right)^{\otimes N}$  is  lower  bounded as
   \begin{align}\label{NN}
   k  \ge   \frac{N^N-1}{N-1} +  1 \,.   
   \end{align}
\end{corollary}
The bound (\ref{NN}) is another example of a non-trivial bound, {\em i.e.} of a bound that improves over bound (\ref{trivial}).  Here the improvement is exponential:  for asymptotically large $N$, the difference between the r.h.s. of Eq.~\eqref{NN}  and the r.h.s. of Eq.~\eqref{trivial} grows as $N^{N-2}$.   It is also worth noting that the bound (\ref{NN}) provides also a small improvement over bound (\ref{old}) in a scenario where all local dimensions are equal. 


\section{Orthogonality graphs for minimal GUPBs}\label{sec:gupbs}

We now derive an analogue of  Proposition~\ref{pro:re} for GUPBs, showing that a certain kind of minimal GUPBs must be associated to regular graphs:

\begin{proposition}\label{pro:regular_graph}
For a minimal  GUPB saturating the bound (\ref{new}),  the  orthogonality graph $G_m$ is a  $(k-\frac{D}{d_m})$-regular graph for every $m\in  \{1,\dots, N\}$.
\end{proposition}

The proof of Proposition~\ref{pro:regular_graph} is provided in Appendix~\ref{appendix:proposition}. We now use Proposition~\ref{pro:regular_graph} to put forward  a possible approach to construct a  GUPB of minimal size and local dimension.   Since UPBs do not exist in $\bbC^2\otimes \bbC^n$ \cite{bennett1999unextendible,divincenzo2003unextendible}, the minimal setting for a GUPB is a three-qutrits system.   By bounds (\ref{old})  and (\ref{new}), we know that the size of a candidate GUPB must be at least $13$. 

Now, Proposition~\ref{pro:regular_graph} shows that, if there exists a GUPB of size $13$ in $\bbC^{3}\otimes \bbC^{3}\otimes \bbC^{3}$, then each orthogonality graph $G_m$ is a $4$-regular graph.    We then have the following proposition:

\begin{proposition}\label{pro:gupb444}
A set of product states  $\{\ket{\varphi_1^{(i)}}_{A_1} \ket{\varphi_2^{(i)}}_{A_2}  \ket{\varphi_3^{(i)}}_{A_3}  \}_{i=1}^{13}$ in $\bbC^3\otimes \bbC^3\otimes \bbC^3$ is a GUPB if the following three conditions hold,
\begin{enumerate}[(i)]
    \item \label{item:graph}$\bigcup_{m=1}^3G_m=K_{13}$, where each orthogonality graph $G_m$ is a $4$-regular graph;
    \item \label{item:choose5} the subspace spanned by any five states in $\{\ket{\varphi_{m}^{(i)}}_{A_{m}}\}_{i=1}^{13}$ has dimension $3$ for any $m=1,2,3$;
    \item \label{item:choose9} the subspace spanned by any nine states in $\{\ket{\varphi_{j_1}^{(i)}}_{A_{j_1}}\otimes \ket{\varphi_{j_2}^{(i)}}_{A_{j_2}}\}_{i=1}^{13}$ has dimension $9$
    for any $(j_1,j_2)\in \{(1,2),(1,3),(2,3)\}$.
\end{enumerate}
\end{proposition}
\begin{proof}
Immediate from  Lemma~\ref{lemma:suf_nec} and the fact that  the orthogonality graphs $(G_m)_{m=1}^3$ are  $4$-regular. 
\end{proof}
\medskip 

Proposition~\ref{pro:gupb444} provides a possible approach  to construct a tripartite GUPB of minimum local dimension. 
There are three steps for constructing a GUPB of size $13$ in $\bbC^3\otimes \bbC^3\otimes \bbC^3$:
\begin{enumerate}
    \item decompose the complete graph $K_{13}$ into three $4$-regular graphs $(G_m)_{m=1}^3$;
\item find an orthogonal representation for each $G_m$ in $\bbC^3$;
\item check the conditions (\ref{item:choose5}) and (\ref{item:choose9}) of Proposition~\ref{pro:gupb444}.
\end{enumerate}

We now discuss the possible ways forward and the challenges arising in the above steps. Regarding step 1, there are many ways to decompose $K_{13}$ into three $4$-regular graphs. In particular, one can decompose $K_{13}$ into three Cayley graphs \cite{graham1995handbook}.
To do this, one has to consider the  group of integers modulo $13$, $\bbZ_{13}=\{0,\pm 1,\pm 2,\pm 3,\pm4, \pm5, \pm 6\}$.   Given a 2-element set $S=\{ p,  q\}$, where  $1\leq p\neq q\leq 6$, one can construct the Cayley graph $G^{(S)}=(V, E)$, where $V=\bbZ_{13}$, and $E=\{(a,b)\mid a-b\in S\cup (-S)\}$. By construction,    $G^{(S)}$ is a $4$-regular graph.  
 By partitioning the set $\{1,2,3,4,5,6\}$  into three 2-element subsets $S_1=\{p_1,  q_1\}$, $S_2=\{p_2,  q_2\}$, and
  $S_3=\{p_3,  q_3\}$,  we then obtain the desired decomposition   $K_{13}=\bigcup_{m=1}^3G^{(S_m)}$. Note that there are $\frac{C_6^2\times C_4^2\times C_2^2} {A_3^3}=15$ distinct partitions of $\{1,2,3,4,5,6\}$  into three 2-element subsets.   Hence,   there are $15$ distinct decompositions of $K_{13}$ into three Cayley graphs.   While step 1 is relative straightforward, a bottleneck arises in step 2, where one has to find an orthogonal representation of the graphs in the decomposition of $K_{13}$.   For each  decomposition,  our algorithm in Sec.~\ref{sec:pri}   can find the orthogonal representations of at most two graphs, leaving the third unspecified.   


The bottleneck of the orthogonal representations remains even if one replaces the decomposition into Cayley graphs with some other decomposition of $K_{13}$ in terms of regular graphs.  
 To better understand the origin of the problem,  we point out that it is not overwhelmingly difficult to find orthogonal representations for all 4-regular graphs with 13 vertices, as the total number of such graphs, up to isomorphism, is  10,880 \cite{meringer1999fast}. However, our algorithm for searching orthogonal representation does not ensure the fulfillment of condition (\ref{item:choose5}) or (\ref{item:choose9}) of Proposition \ref{pro:gupb444}. By iterating the algorithm on the same graph, one may hope to  find orthogonal representations fulfilling condition (\ref{item:choose5}) by chance. Unfortunately, we did not encounter any such solutions.
 One way to circumvent the problem would be to translate the condition (\ref{item:choose5}) into a constraint that has to be satisfied while searching for the orthogonal representation with our algorithm in Sec.~\ref{sec:pri}. The problem with this approach is that condition (\ref{item:choose5})  results in non-linear constraints, which heavily slow down the convergence of the optimization process. 
 We managed to run the modified algorithm twice on all 4-regular 13-vertex graphs, but did not find any orthogonal representation satisfying condition (\ref{item:choose5}).
 Due to these obstacles, finding an example of GUPB through the above route still requires a major investment of computational resources.

\section{conclusions}\label{sec:con}
In this paper, we  established a characterization of UPBs and GUPBs in terms of orthogonality graphs. Building on this characterization, we developed a condtructing method for finding UPBs in low dimensional systems, and we derived a new lower bound on the number of elements in any GUPB. Our bound  significantly improves on the state of the art \cite{PhysRevA.106.012442}, thus placing stronger restrictions on potential candidates of GUPBs.  Equivalently, our bound implies an upper  bound on the   rank of any  bound entangled state built from a GUPB.       Our results indicate a potential route to find a minimal tripartite GUPB consisting of  13 product vectors.  While the numerical search for such GUPB is still challenging, our construction helps clarify where the problems lie, and may eventually help find a suitable modification that is amenable to numerical search.    

Besides addressing the open problem of the existence of GUPBs,  we provided a systematic route to the construction of multipartite UPBs of any desired size between the minimum and the maximum.  
  Equivalently, our construction can be viewed as a systematic way of constructing bound entangled states of different ranks.    In addition,  our results have an application to the study of nonlocality without entanglement.  In Ref.  \cite{vsupic2022self}  it was shown that  quantum measurements exhibiting nonlocality without entanglement can be certified in a device-independent way.   Since our results provide a  systematic construction of multipartite  UPBs, the corresponding scenarios of nonlocality without entanglement are likely to  give rise to  new self-testing procedures.  Finally, another interesting direction is the study of non-trivial Bell inequalities with no quantum violation \cite{augusiak2011bell, Augusiak2012tight,fritz2013local}.  In this context, our work can be used to construct  such inequalities in multipartite systems with larger local dimensions,  going beyond the multiqubit scenario typically considered in the literature.

\section*{Acknowledgments}
\label{sec:ack}		
We thank Lin Chen and Yiwei Zhang for discussing this problem. F.S., G.B. and G.C. acknowledge funding from
the Hong Kong Research Grant Council through Grants No.
17300918 and No. 17307520, and through the Senior Research
Fellowship Scheme SRFS2021-7S02. 
This publication was made possible through the support of the ID\# 62312 grant
from the John Templeton Foundation, as part of the ‘The Quantum Information Structure of Spacetime’
  Project (QISS). The opinions expressed in this project are those of the
authors and do not necessarily reflect the views of the John Templeton Foundation. Research at the Perimeter Institute is supported by the
Government of Canada through the Department of Innovation, Science and Economic Development Canada and by the
Province of Ontario through the Ministry of Research, Innovation and Science.  X.Z. acknowledges funding from the NSFC under Grants No. 12171452 and No. 11771419,
the Anhui Initiative in Quantum Information Technologies under Grant No. AHY150200, the Innovation Program for Quantum
Science and Technology (2021ZD0302904),  and the National Key
Research and Development Program of China (2020YFA0713100).  Q.Z. acknowledges funding from the HKU Seed Fund for New Staff.

\bibliographystyle{IEEEtran}
\bibliography{reference}

\appendix

\section{Proofs of Lemmas~\ref{lemma:suf_nec}, \ref{lem:unsaturated},  and~\ref{lem:deg}  
}\label{appendix:lemma}
\setcounter{lemma}{0}
\begin{lemma}
Let   $\cU$ be a set of $k$ product vectors in $\bbC^{d_1}\otimes \bbC^{d_2}\otimes\cdots\otimes \bbC^{d_N}$, and let $(G_m)_{m=1}^N$ be the corresponding orthogonality graphs.   The set $\cU$ is a UPB if and only if the following conditions hold:
\begin{enumerate}[(i)]
    \item $\bigcup_{m=1}^NG_m=K_k$;
    \item  $\bigcup_{m=1}^NW_m\neq V$ for every $N$-tuple  $(W_1 ,W_2,\ldots,W_N) $ in which  $W_m$ is an unsaturated set for  $G_m$ for  every $m\in  \{1,
    \dots, N\}$.  \end{enumerate}
\end{lemma}
\begin{proof}
 The proof builds on arguments  by Bennett {\em et al.} (cf. Lemma 1 of~\cite{bennett1999unextendible}), which are translated here into the graph theoretic framework of our paper by using the notion of saturated set.    First,  we observe that  the product states in the set $\cU$  are mutually orthogonal  if and only if $\bigcup_{m=1}^NG_m=K_k$. Hence, we only need to show that a set of orthogonal product states $\cU$  is unextendible if and only if condition (ii) holds. 
 
 The ``only if'' part is proven by contrapositive: we show that if condition (ii) is violated, then the set $\cU$ must be extendible. The proof is as follows:  If there exists an unsaturated set $W_m$ of $G_m$ for every $1\leq m\leq N$ such that $\bigcup_{m=1}^NW_m=V$, then we can find a state $\ket{\psi}_{A_m} \in \bbC^{d_m}$ that is orthogonal to any state in $\{\ket{\varphi_m^{(i)}}_{A_m}\mid v_i\in W_m\}$ for every $1\leq m\leq N$, and $\ket{\psi_1}_{A_1}\ket{\psi_2}_{A_2}\cdots \ket{\psi_N}_{A_N}$ is orthogonal to any state in $\cU$. 
 
For the ``if'' part, we also proceed by contrapositive: we  assume that $\cU$ is extendible and show that condition (ii) must be violated.  If $\cU$ is extendible, then there exists a product state $\ket{\psi_1}_{A_1}\ket{\psi_2}_{A_2}\cdots \ket{\psi_N}_{A_N}$ that is orthogonal to any state in $\cU$. Let $W_m=\{v_i\mid \braket{\psi_m}{\varphi_m^{(i)}}_{A_m}=0\}$ for every $1\leq m\leq N$, then $W_m$ must be an  unsaturated set of $G_m$ for every $1\leq m\leq N$, and $\bigcup_{m=1}^NW_m=V$.
\end{proof}

\begin{figure*}[t]
		\centering
		\includegraphics[scale=0.6]{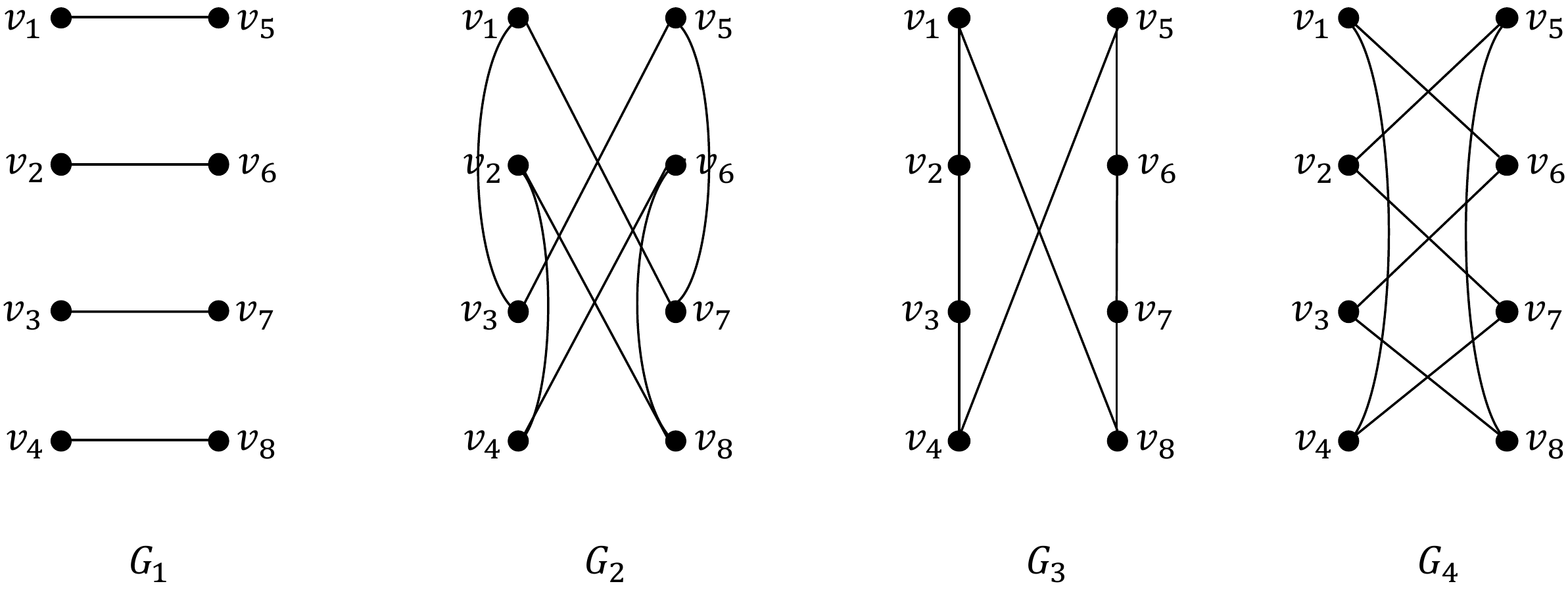}
		\caption{Orthogonality graphs of the UPB  in Example \ref{example:upb2233}. } \label{fig:orthogonality_graph_2233}
	\end{figure*}

\begin{lemma}
 Let $(G_m)_{m=1}^N$ be the orthogonality graphs associated to a UPB of size $k$.  Then, the size of any unsaturated set $W_m$ in $G_m$   is upper bounded as 
 \begin{align}
 |W_m|  \le  k   -1-   \sum_{i\in  \{1,\dots,  N\}  \setminus \{m\} }  \,  (d_i-1) \,.
 \end{align}
 \end{lemma}
\begin{proof}
  The proof is by contradiction.  Suppose that there existed an integer $m_0\in  \{1,\dots, N\}$ and an unsaturated set $W_{m_0}$ such that $|W_{m_0}|\ge  k   -   \sum_{i \not  =  m_0}   \,  (d_i-1) $.     Then, the set $V\setminus  W_{m_0}$ contains  $l\le \sum_{i \not  =  m_0}   \,  (d_i-1)$ vertices.   These $l$ vertices can be divided into $N-1$ subsets, putting at most $d_m-1$ vertices in the $m$-th subset, for every $m  \in  \{1,\dots, N\}\setminus  \{m_0\}$. 
  The $m$-th subset, denoted by $W_m$,  is by construction an unsaturated set in $G_m$.   Also, the above construction guarantees that $\bigcup_{m=1}^N  W_m  =  V$.  But   this condition is in contradiction with the fact that the graphs $(G_m)_{m=1}^N$ are the orthogonality graphs of a UPB, because  Lemma \ref{lemma:suf_nec} showed the relation $\bigcup_{m=1}^N  W_m \not =  V$ for every  $N$-tuple of unsaturated subsets  $(W_1,\dots,  W_N)$.  This concludes the proof by contradiction.  
  \end{proof}

\begin{lemma}
For every  UPB in   $\bbC^{d_1}\otimes \bbC^{d_2}\otimes\cdots\otimes \bbC^{d_N}$, the degrees of the vertices in the  orthogonality graphs $(G_m)_{m=1}^N$ must satisfy the condition 
\begin{equation}
\begin{aligned}
  &d_m-1\leq {\rm deg}_{G_m}(v_i)\leq k-1-  \sum_{i \in  \{1,\dots,  N\}\setminus \{m\}}(d_i-1),\\
  &\qquad \forall v_i\in  V \, ,\quad \forall  m\in  \{1,\dots,  N\} \,.
  \end{aligned}
\end{equation}
\end{lemma}
\begin{proof}
The upper bound is immediate  from the fact that the degree ${\rm deg}_{G_m}(v_i)  =  |N_{G_m}(v_i) |$, where $N_{G_m}(v_i)$  is the neighborhood of $v_i$ in $G_m$.  Since the neighborhood of a vertex in an orthogonality graph is, by definition, an unsaturated set, the upper bound on its size follows from Lemma \ref{lem:unsaturated}. 

For the lower bound, we assume that there exists an orthogonality graph $G_m=(V,E_m)$ and a vertex $v_j\in V$ such that ${\rm deg}_{G_m}(v_j)\leq d_m-2$. Then we can find a state $\ket{\phi}_{A_m}$ in $\bbC^{d_m}$ which is orthogonal to any state in $\{\ket{\varphi_m^{(i)}}_{A_m}\mid v_i\in \{v_j\}\cup N_{G_m}(v_j)\}$. The product state $\ket{\varphi_1^{(j)}}_{A_1} \cdots \ket{\varphi_{m-1}^{(j)}}_{A_{m-1}} \ket{\phi}_{A_m}\ket{\varphi_{m+1}^{(j)}}_{A_{m+1}}\cdots \ket{\varphi_N^{(j)}}_{A_N}$ is orthogonal to any state in $\{\ket{\varphi_1^{(i)}}_{A_1} \ket{\varphi_2^{(i)}}_{A_2} \cdots \ket{\varphi_N^{(i)}}_{A_N}  \}_{i=1}^{k}$, which contradicts that $\{\ket{\varphi_1^{(i)}}_{A_1} \ket{\varphi_2^{(i)}}_{A_2} \cdots \ket{\varphi_N^{(i)}}_{A_N}  \}_{i=1}^{k}$ is a UPB. 
\end{proof}

\begin{example}\label{example:upb2233}
    The following product vectors form a UPB of size 8 in $\bbC^2\otimes \bbC^2\otimes \bbC^3 \otimes \bbC^3$:
    \begin{widetext}
\begin{equation}\label{eq:upb2233}
    \begin{aligned}
    \ket{\psi_1}=&\ket{0}_{A_1}     \ket{0}_{A_2} (\ket{0}+\ket{1}+\ket{2})_{A_3} (\ket{0}+\ket{1}+\ket{2})_{A_4},\\  
    \ket{\psi_2}=&(\ket{0}+\ket{1})_{A_1} (\ket{0}+\ket{1})_{A_2} (\ket{0}+\ket{1}-2\ket{2})_{A_3}(2\ket{0}-\ket{1}-2\ket{2})_{A_4},  \\
    \ket{\psi_3}=&(\ket{0}+2\ket{1})_{A_1} \ket{1}_{A_2} (4\ket{0}+2\ket{1}+3\ket{2})_{A_3}(3\ket{0}+6\ket{1}-2\ket{2})_{A_4}, \\ 
    \ket{\psi_4}=&(\ket{0}+3\ket{1})_{A_1} (\ket{0}-\ket{1})_{A_2}(2\ket{0}-\ket{1}-2\ket{2})_{A_3}(\ket{0}+\ket{1}-2\ket{2})_{A_4},\\
    \ket{\psi_5}=&\ket{1}_{A_1} \ket{0}_{A_2}(\ket{0}+4\ket{1}-\ket{2})_{A_3}(\ket{0}+4\ket{1}-\ket{2})_{A_4},\\ 
    \ket{\psi_6}=&(\ket{0}-\ket{1})_{A_1}(\ket{0}+\ket{1})_{A_2}(2\ket{0}+\ket{1}+6\ket{2})_{A_3}(-8\ket{0}+5\ket{1}+3\ket{2})_{A_4},\\
   \ket{\psi_7}=&(2\ket{0}-\ket{1})_{A_1}\ket{1}_{A_2} (3\ket{0}+6\ket{1}-2\ket{2})_{A_3}(4\ket{0}+2\ket{1}+3\ket{2})_{A_4},\\ 
    \ket{\psi_8}=&(3\ket{0}-\ket{1})_{A_1}(\ket{0}-\ket{1})_{A_2}(-8\ket{0}+5\ket{1}+3\ket{2})_{A_3}(2\ket{0}+\ket{1}+6\ket{2})_{A_4}.
    \end{aligned}
\end{equation}
   \end{widetext}  
\end{example}
\begin{proof}
The   orthogonality graphs $(G_m)_{m=1}^4$ of $\{\ket{\psi_i}\}_{i=1}^8$ is showed in  Fig.~\ref{fig:orthogonality_graph_2233}. Then it is easy to check that $\cup_{m=1}^4 G_m=K_8$. 
Hence, the first condition in Lemma~\ref{lemma:suf_nec} is satisfied. 
Regarding the second condition, note that every two vectors in  the set $\{\ket{0}, (\ket{0}+\ket{1}), (\ket{0}+2\ket{1}), (\ket{0}+3\ket{1}), \ket{1}, (\ket{0}-\ket{1}), (2\ket{0}-\ket{1}), (3\ket{0}-\ket{1})\}\subset \bbC^2$ are linearly independent.  Hence,  the size of any unsaturated set $W_1$ in $G_1$ can be at most 1. Since every three vectors in  the set $\{\ket{0}, (\ket{0}+\ket{1}), \ket{1}, (\ket{0}-\ket{1}), \ket{0}, (\ket{0}+\ket{1}), \ket{1}, (\ket{0}-\ket{1})\}\subset \bbC^2$ are linearly independent,  the size of any unsaturated set $W_2$ in $G_2$ can be at most 2. Moreover, since any three vectors in the set  $\{(\ket{0}+\ket{1}+\ket{2}) , (\ket{0}+\ket{1}-2\ket{2}), (4\ket{0}+2\ket{1}+3\ket{2}), (2\ket{0}-\ket{1}-2\ket{2}), (\ket{0}+4\ket{1}-\ket{2}), (2\ket{0}+\ket{1}+6\ket{2}), (3\ket{0}+6\ket{1}-2\ket{2}), (-8\ket{0}+5\ket{1}+3\ket{2}) \}\subset \bbC^3$ are linearly independent,  the size of any unsaturated set $W_m$ in $G_m$ is at most 2 for every $3\leq m\leq 4$. Putting everything together, we obtain that the union of  any four  unsaturated sets $(W_m)_{m=1}^4$ cannot contain all vertices in $V$. Since both conditions in Lemma~\ref{lemma:suf_nec} are satisfied,  we conclude that the vectors  $\{\ket{\psi_i}\}_{i=1}^8$ form a UPB.

\end{proof}

\section{Proofs of Propositions~\ref{prop:boundcomparison},~\ref{cor:p},~\ref{pro:shakinghand}, and \ref{pro:regular_graph}}\label{appendix:proposition}
\setcounter{proposition}{1}

 \begin{proposition}
The r.h.s. of Eq.~\eqref{eq:lower_bound} is always larger than or equal to the r.h.s. of Eq.~\eqref{old}.
 \end{proposition}
\begin{proof}
 First, we show that r.h.s. of Eq.~\eqref{eq:lower_bound} satisfies the equality 
 \begin{align}\label{eq:equal}
 \fc{\sum_{m=1}^N\frac{D}{d_m}-1}{N-1}=\fl{\sum_{m=1}^N\frac{D}{d_m}-2}{N-1}+1 \,.
 \end{align}
 If the number $s:  =\frac{\sum_{m=1}^N\frac{D}{d_m}-1}{N-1}$ is an integer, 
then one has the relations $s  = \fc{\sum_{m=1}^N\frac{D}{d_m}-1}{N-1}$ and $\fl{\sum_{m=1}^N\frac{D}{d_m}-2}{N-1}+1=\fl{s(N-1)-1}{N-1}+1=s$. If instead $\frac{\sum_{m=1}^N\frac{D}{d_m}-1}{N-1}$ is not an integer, it can be written as $\frac{\sum_{m=1}^N\frac{D}{d_m}-1}{N-1}=s+\frac{t}{N-1}$ for some integer $s$ and some integer $1\leq t\leq N-2$. Then, $\fc{\sum_{m=1}^N\frac{D}{d_m}-1}{N-1}=s+1$, and  $\fl{\sum_{m=1}^N\frac{D}{d_m}-2}{N-1}+1=s+\fl{t-1}{N-1}+1=s+1$. Thus, $k\geq \fc{\sum_{m=1}^N\frac{D}{d_m}-1}{N-1}=\fl{\sum_{m=1}^N\frac{D}{d_m}-2}{N-1}+1$. 

To conclude, we use the bound 
\begin{equation}\label{bb}
 \begin{aligned} &\fl{\sum_{m=1}^N\frac{D}{d_m}-2}{N-1}+1\geq \fl{N\frac{D}{d_{\text{max}}}-2}{N-1}+1\\
 &=\frac{D}{d_{\text{max}}}+\fl{\frac{D}{d_{\text{max}}}-2}{N-1}+1 \,,
 \end{aligned} 
 \end{equation}
 where the last term in the inequality is the r.h.s. of Eq.~\eqref{old}.   Combining Eqs.~\eqref{eq:equal} and (\ref{bb}) we then obtain that the r.h.s. of Eq.~\eqref{eq:lower_bound} is larger than or equal to the r.h.s. of Eq.~\eqref{old}. 
 \end{proof}

\begin{proposition}
In the tripartite case, the  bound  (\ref{new})  is  non-trivial when  
 $(d_1,d_2,d_3)= (2p,2p,3p-1)$ for some integer $p\ge 2$, and when $(d_1,d_2,d_3)= (2p-1,\tilde{d},3p-2)$ for some integer $p\geq 2$ and some integer $\tilde d  \in  [2p-1,  3p-2]$. 
\end{proposition}
\begin{proof}
  In both cases, we prove the inequality 
  \begin{align}\label{toprove}\fl{\sum_{m=1}^3\frac{D}{d_m}-2}{N-1}+1\geq d_1+\frac{D}{d_1}+1\,.
  \end{align}
When $(d_1,d_2,d_3)= (2p,2p,3p-1)$ with integer $p\geq 2$, we have
\begin{equation*}
       \fl{\sum_{m=1}^3\frac{D}{d_m}-2}{N-1}+1  = 8p^2-2p, 
\end{equation*}
and \begin{equation*}
     d_1+\frac{D}{d_1}+1=6p^2+1.
\end{equation*}
Hence, the inequality (\ref{toprove}) is equivalent to 
\begin{equation}\label{eq:geq_0}
    2p^2-2p-1\geq 0.
\end{equation}
Note that Eq.~\eqref{eq:geq_0} holds for any $p\geq 2$.  Hence, the bound (\ref{new}) is non-trivial  for every $p\geq 2$.

Let us now consider the case where $(d_1,d_2,d_3)= (2p-1,\tilde{d},3p-2)$ for some integer $p\geq 2$ and some integer $\tilde d$  satisfying $2p-1\leq \tilde{d}\leq 3p-2$.  In this case, we have 
\begin{equation*}
\fl{\sum_{m=1}^3\frac{D}{d_m}-2}{N-1}+1  =   \fl{\tilde{d}(5p-3)+6p^2-7p}{2}+1, 
\end{equation*}
 and 
\begin{equation*}
      d_1+\frac{D}{d_1}+1=2p+\tilde{d}(3p-2), 
\end{equation*}
and the inequality (\ref{toprove}) is equivalent to 
\begin{equation}\label{eq:ceil}
\left\lceil\frac{\tilde{d}(p-1)-6p^2+11p-2}{2}\right\rceil\leq 0.
\end{equation}
Notice that if Eq.~\eqref{eq:ceil} holds for  $\tilde{d}=3p-2$, then it  holds for any $2p-1\leq \tilde{d}\leq 3p-2$.
Since
\begin{equation}
   (3p-2)(p-1)-6p^2+11p-2=-3p^2+6p\leq 0
\end{equation}
for any $p\geq 2$, then Eq.~\eqref{eq:ceil} holds for any $p\geq 2$ and $2p-1\leq \tilde{d}\leq 3p-2$.  Therefore,  Eq. (\ref{toprove}) holds for any $p\geq 2$ and $2p-1\leq \tilde{d}\leq 3p-2$, meaning that the bound (\ref{new}) is non-trivial for these values. 
\end{proof}

\medskip

We now provide the proofs of  Propositions \ref{pro:shakinghand} and \ref{pro:regular_graph}.  The proofs are presented in inverted order, because the proof of Proposition \ref{pro:shakinghand} uses Proposition \ref{pro:regular_graph} as an intermediate step.

\setcounter{proposition}{4}
\begin{proposition}
For a minimal  GUPB saturating the bound (\ref{new}),  the  orthogonality graph $G_m$ is a  $(k-\frac{D}{d_m})$-regular graph for every $m\in  \{1,\dots, N\}$.
\end{proposition}
\begin{proof}
For the bound (\ref{new}) to be saturated, we must have $k =  \frac{\sum_{m=1}^N\frac{D}{d_m}-1}{N-1}$, or equivalently,    
\begin{align}\label{minimality}
(N-1) \,k =  {\sum_{m=1}^N\frac{D}{d_m}-1} \,.
\end{align}
Since $\{\ket{\varphi_1^{(i)}}_{A_1} \ket{\varphi_2^{(i)}}_{A_2} \cdots \ket{\varphi_N^{(i)}}_{A_N}  \}_{i=1}^k$ is still a UPB in the bipartition $A_m\mid  \{A_1A_2\cdots A_N\}\setminus\{A_m\}$ for $1\leq m\leq N$, then by Lemma~\ref{lem:deg}, the degree of every vertex $v_i$ of $G_m$  satisfies
\begin{equation}\label{eq:degvi_kminusDdm}
    {\rm deg}_{G_{m}}(v_i)\leq k-\frac{D}{d_m}.
\end{equation}
On the other hand, the minimality condition (\ref{minimality}) implies
\begin{equation}\label{eq:k_1}
\sum_{m=1}^N(k-\frac{D}{d_m})= k-1  \leq \sum_{m=1}^N{\rm deg}_{G_{m}}(v_i),
\end{equation}
where the inequality comes from  $\bigcup_{m=1}^NG_m = K_k$. Combining the two inequalities above, we obtain ${\rm deg}_{G_{m}}(v_i)= k-\frac{D}{d_m}$ for $1\leq m\leq N$.
\end{proof}

\setcounter{proposition}{3}
\begin{proposition}
  If  at least one of the local dimensions $(d_m)_{m=1}^N$ is even and the sum $\sum_{m=1}^N\frac{D}{d_m}-1$ is an odd multiple of $N-1$, then  the size of any GUPB in $\bbC^{d_1}\otimes \bbC^{d_2}\otimes\cdots\otimes \bbC^{d_N}$ is lower bounded as
\begin{equation}
       k\geq \frac{\sum_{m=1}^N\frac{D}{d_m}-1}{N-1}+1.
\end{equation}
\end{proposition}
\begin{proof}
  We prove that, when the local dimensions satisfy the above conditions, the bound (\ref{new}) cannot hold with the equality sign.     
  If a GUPB saturated the bound (\ref{new}), its size should be   $k'=\frac{\sum_{m=1}^N\frac{D}{d_m}-1}{N-1}$ and $k'$ is odd.  
  This means that each orthogonality graph $G_m$ is a $(k'-\frac{D}{d_m})$-regular graph with $k'$ vertices  by Proposition~\ref{pro:regular_graph}. By Handshaking Lemma \cite{gunderson2010handbook},  $k'(k'-\frac{D}{d_m})$ must be even for each $1\leq m\leq N$. Assume $d_{m}$ is even, then $\frac{D}{d_{m'}}$ is even for $m'\neq m$. In this case, $k'(k'-\frac{D}{d_{m'}})$ is odd, and  this is impossible. Thus a GUPB of size  $k'$ does not exist. This completes the proof. 
\end{proof}

\end{document}